\documentclass[10pt,a4paper]{article}

\usepackage{amsmath}
\usepackage{amssymb}
\usepackage{color}
\usepackage{bbm}
\usepackage{mathtools}
\usepackage{transparent}
\usepackage{color}
\usepackage{graphicx}
\usepackage{tikz}
\usepackage{authblk}
\usepackage{float}

\usepackage{amsthm}
\newcounter{foo}

\newtheorem{Lemma}[foo]{Lemma}
\newtheorem{Proposition}[foo]{Proposition}
\newtheorem{Definition}{Definition}

\def\url#1{{\tt #1}}

\def\tr{\operatorname{tr}}
\def\idty{{\mathbbm 1}} 

\def\Rl{{\mathbb R}}\def\Cx{{\mathbb C}}

\def\brAAket#1#2#3{\langle#1\vert#2\vert#3\rangle}

\def\ket #1{\vert#1\rangle}
\def\ketbra #1#2{{\vert#1\rangle\!\langle#2\vert}}
\def\kettbra#1{\ketbra{#1}{#1}}

\def\abs#1{\vert#1\vert}

\let\veps\varepsilon

\def\error{\varepsilon}
\def\err#1{\error(A'_{#1}|A_{#1})}    
\def\errM#1{\error_M(A'_{#1}|A_{#1})} 
\def\errC#1{\error_C(A'_{#1}|A_{#1})} 
\def\errE#1{\error_E(A'_{#1}|A_{#1})} 
\def\errA#1{\error_D(A'_{#1}|A_{#1})} 
\def\errL#1{\error_L(A'_{#1}|A_{#1})} 

\def\UR#1{{\mathcal U}_{#1}} 
\def\errorL#1{\error_L(#1|A)} 
\def\cB{{\mathcal B}}  
\def\ii{{\mathbbm 1}} 
\def\ops{{\mathcal S}}
\def\maxEV{\lambda_{\rm max}}

\def\cc{\check c} 
\def\trans{^{\sf T}}

\def\boxformula#1{\fbox{{$\displaystyle#1$}}}

\title{Measurement Uncertainty for Finite Quantum Observables}

\author{Ren\'e Schwonnek\footnote{rene.schwonnek@itp.uni-hannover.de}, 
       David Reeb\footnote{david.reeb@itp.uni-hannover.de}, and 
       Reinhard F. Werner\footnote{reinhard.werner@itp.uni-hannover.de}}

\affil{Quantum Information Group, Institute for Theoretical Physics,\\ Leibniz Universit\"at Hannover}

\date{April 1, 2016}

\begin{document}

\maketitle

\begin{abstract}Measurement uncertainty relations are lower bounds on the errors of any approximate joint measurement of two or more quantum observables. The aim of this paper is to provide methods to compute optimal bounds of this type. The basic method is semidefinite programming, which we apply to arbitrary finite collections of projective observables on a finite dimensional Hilbert space. The quantification of errors is based on an arbitrary cost function, which assigns a penalty to getting result $x$ rather than $y$, for any pair $(x,y)$. This induces a notion of optimal transport cost for a pair of probability distributions, and we include an appendix with a short summary of optimal transport theory as needed in our context. There are then different ways to form an overall figure of merit from the comparison of distributions. We consider three, which are related to different physical testing scenarios. The most thorough test compares the transport distances between the marginals of a joint measurement and the reference observables for every input state. Less demanding is a test just on the states for which a ``true value'' is known in the sense that the reference observable yields a definite outcome. Finally, we can measure a deviation as a single expectation value by comparing the two observables on the two parts of a maximally entangled state. All three error quantities have the property that they vanish if and only if the tested observable is equal to the reference. The theory is illustrated with some characteristic examples.\end{abstract}

\section{Introduction}

Measurement uncertainty relations are quantitative expressions of complementarity. As Bohr often emphasized, the predictions of quantum theory are always relative to some definite experimental arrangement, and these settings  often exclude each other. In particular, one has to make a choice of measuring devices, and typically quantum observables cannot be measured simultaneously. This often used term is actually misleading, because time has nothing to do with it. For a better formulation recall that quantum experiments are always statistical, so the predictions refer to the frequency with which one will see certain outcomes when the whole experiment is repeated very often. So the issue is not {\it simultaneous} measurement of two observables, but {\it joint} measurement in the same shot. That is, a device $R$ is a joint measurement of observable $A$ with outcomes $x\in X$ and observable $B$ with outcomes $y\in Y$, if it produces outcomes of the form $(x,y)$ in such a way that if we ignore outcome $y$, the statistics of the $x$ outcomes is always (i.e., for every input state) the same as obtained with a measurement of $A$, and symmetrically for ignoring $x$ and comparing with $B$. It is in this sense that non-commuting projection valued observables fail to be jointly measurable.

However, this is not the end of the story. One is often interested in {\it approximate} joint measurements. One such instance is Heisenberg's famous $\gamma$-ray microscope \cite{Heisenberg}, in which a particle's position is measured by probing it with light of some wavelength $\lambda$, which from the outset sets a scale for the accuracy of this position measurement. Naturally, the particle's momentum is changed by the Compton scattering, so if we make a momentum measurement on the particles after the interaction, we will find a different distribution from what would have been obtained directly. Note that in this experiment  we get from every particle a position value and momentum value. Moreover, errors can be quantified by comparing the respective distributions with some ideal reference: The {\it accuracy} of the microscope position measurement is judged by the degree of agreement between the distribution obtained and the one an ideal position measurement would give. Similarly, the {\it disturbance} of momentum is judged by comparing a directly measured distribution with the one after the interaction. The same is true for the {\it uncontrollable disturbance} of momentum. This refers to a scenario, where we do not just measure momentum after the interaction, but try to build a device that recovers the momentum in an optimal way, by making an arbitrary measurement on the particle after the interaction, utilizing everything that is known about the microscope, correcting all known systematic errors, and even using the outcome of the position measurement. The only requirement is that at the end of the experiment, for each individual shot, some value of momentum must come out. Even then it is impossible to always reproduce the pre-microscope distribution of momentum.
The tradeoff between accuracy and disturbance is quantified by a measurement uncertainty relation. Since it simply quantifies the impossibility of a joint exact measurement, it simultaneously gives bounds on how an approximate momentum measurement irretrievably disturbs position. The basic setup is shown in Fig.\ \ref{fig:setup}.

\begin{figure}[H]
\centering
\def\svgwidth{0.7\textwidth} \input{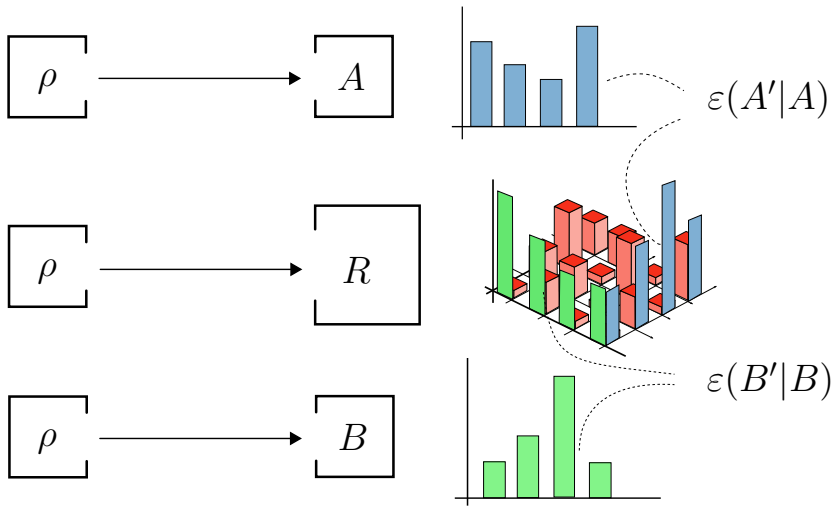}
\caption{Basic setup of measurement uncertainty relations. The approximate joint measurement $R$ is shown in the middle, with its array of output probabilities. The marginals $A'$ and $B'$ of this array are compared with the output probabilities of the reference observables $A$ and $B$, shown at the top and at the bottom. The uncertainties $\error(A'|A)$ and $\error(B'|B)$ are quantitative measures for the difference between these distributions.\label{fig:setup}}
\end{figure}

Note that in this description of errors we did not ever bring in a comparison with some hypothetical ``true value''. Indeed it was noted already by Kennard \cite{Kennard} that such comparisons are problematic in quantum mechanics. Even if one is willing to feign hypotheses about the true value of position, as some hidden variable theorists will, an operational criterion for agreement will always have to be based on statistical criteria, i.e., the comparison of distributions. Another fundamental feature of this view of errors is that it provides a figure of merit for the comparison of two devices, typically some ideal reference observable and and an approximate version of it. An ``accuracy'' $\veps$ in this sense is a promise that no matter which input state is chosen, the distributions will not deviate by more than $\veps$. Such a promise does not involve a particular state. This is in contrast to {\it preparation uncertainty} relations, which quantify the impossibility to find a state for which the distributions of two given observables (e.g., position and momentum) are both sharp.

Measurement uncertainty relations in the sense described here were first introduced for position and momentum in \cite{Wer04}, and were initially largely ignored. A bit earlier, an attempt by Ozawa \cite{Ozawa} to quantify error-disturbance tradeoffs with state dependent and somewhat unfortunately chosen \cite{BLW-rms} quantities had failed, partly for reasons already pointed out in \cite{Appleby}. When experiments confirmed some predictions of the Ozawa approach (including the failure of the error-disturbance tradeoff), a debate ensued \cite{BLW1,OzawaDis,BLW2,ApplebyNew}. Its unresolved part is whether a meaningful role for Ozawa's definitions can be found. Technically, the computation of measurement uncertainty for position and momentum in \cite{BLW2} carries over immediately to more general phase spaces \cite{WerFou,BKW}. Apart from some special further computed instances \cite{AMU,BLW2014}, this remained the only case in which sharp measurement uncertainty relations could be obtained. This was in stark contrast with preparation uncertainty, for which an algorithm based on solving ground state problems \cite{AMU} efficiently provides the optimal relations for generic sets of observables.  The main aim of the current paper is to provide efficient algorithms also for sharp measurement uncertainty relations.

In order to do that we restrict the setting in some ways, but allow maximal generality in others. We will restrict to finite dimensional systems, and reference observables which are projection valued and non-degenerate. Thus, each of the ideal observables will basically be given by an orthonormal basis in the same $d$-dimensional Hilbert space. The labels of this basis are the outcomes $x\in X$ of the measurement, where $X$ is a set of $d$ elements. We could choose all $X=\{1,\ldots,d\}$, but it will help to keep track of things using a separate set for each observable. Moreover, this includes the choice $X\subset\Rl$, the set of eigenvalues of some hermitian operator. We allow not just two observables but any finite number $n\geq2$ of them. This is makes some expressions easier to write down, since the sum of an expression involving observable $A$ and analogous one for observable $B$ becomes an indexed sum. We also allow much generality in the way errors are quantified.  In earlier works, we relied on two elements to be chosen for each observable, namely a metric $D$ on the outcome set, and an error exponent $\alpha$, distinguishing, say absolute ($\alpha=1$), root-mean-square ($\alpha=2$), and maximal ($\alpha=\infty$) deviations. Deviations were then averages of $D(x,y)^\alpha$. Here we generalize further to an arbitrary {\it cost function} $c:X\times X\to\Rl$, which we take to be positive, and zero exactly on the diagonal (e.g., $c(x,y)= D(x,y)^\alpha$), but not necessarily symmetric. Again this generality comes mostly as a simplification of notation. For a reference observable $A$ with outcome set $X$ and an approximate version $A'$ with the same outcome set, this defines an error $\err{}$. Our aim is to provide algorithms for computing the {\it uncertainty diagram} associated with such data, of which Fig.~\ref{fig:URs} gives an example. The given data for such a diagram are $n$ projection valued observables $A_1,\ldots,A_n$, with outcome sets $X_i$, for each of which we are given also a cost function $c_i:X_i\times X_i\to\Rl$ for quantifying errors. An approximate joint measurement is then an observable $R$ with outcome set $\bigtimes_i X_i$, and hence with POVM elements $R(x_1,\ldots,x_n)$, where $x_i\in X_i$. By ignoring every output but one we get the $n$ marginal observables
\begin{equation}\label{marginal}
  A'_i(x_i)=\sum_{x_1,\ldots,x_{i-1},x_{i+1},\ldots,x_n} R(x_1,\ldots,x_n)
\end{equation}
and a corresponding tuple
\begin{equation}\label{errortuple}
  \vec\error(R)=\bigl(\err1,\ldots, \err n\bigr)
\end{equation}
of errors. The set $\UR L$ of such tuples, as $R$ runs over all joint measurements, is the \emph{uncertainty region}. The surface bounding this set from below describes the uncertainty tradeoffs. For $n=2$ we call it the tradeoff curve.
Measurement uncertainty is the phenomenon that, for general reference observables $A_i$, the uncertainty region is bounded away from the origin. In principle there are many ways to express this mathematically, from a complete characterization of the exact \textit{tradeoff} curve, which is usually hard to get, to bounds which are simpler to state, but suboptimal. Linear bounds will play a special role in this paper.

\begin{figure}[htb]
\centering \def\svgwidth{0.5\textwidth} \input{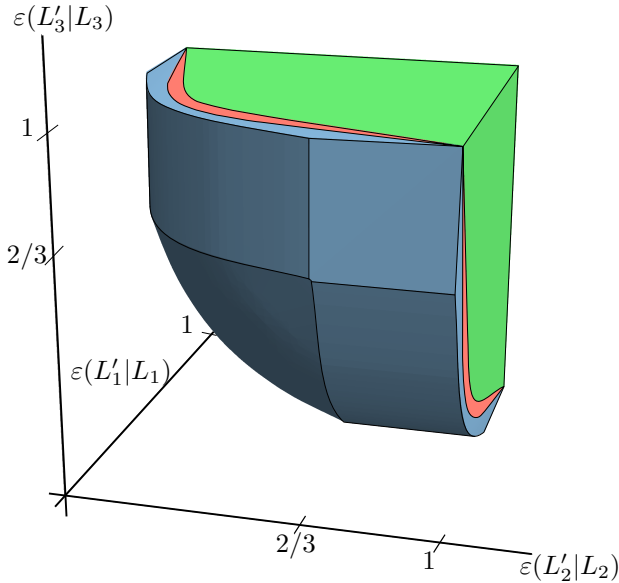}
 \caption{Uncertainty regions for three reference observables, namely the angular momentum
components $L_1,L_2,L_3$ for spin $1$, each with outcome set $X=\{-1,0,+1\}$ and the choice
$c(x,y)=(x-y)^2$ for the cost function. The three regions indicated correspond to the different
overall figures of merit
$\errM{}\geq\errC{}\geq\errE{}$ described in Sect.~\ref{sec:measures}.
\label{fig:URs}}
\end{figure}

We will consider three ways to build a single error quantity out of the comparison of distributions, denoted by $\errM{}$, $\errC{}$, and $\errE{}$. These will be defined in Sect.~\ref{sec:measures}. For every choice of observables and cost functions, each will give an uncertainty region, denoted  by $\UR M$, $\UR C$, and $\UR E$, respectively.  Since the errors are all based on the same cost function $c$, they are directly comparable (see Fig.~\ref{fig:URs}). We show in Sect.~\ref{sec:convex} that the three regions are convex, and hence characterized completely by linear bounds. In Sect.~\ref{sec:sdp} we show how to calculate the optimal linear lower bounds by semidefinite programs. Finally, an Appendix collects the basic information on the beautiful theory of optimal transport, which is needed in Sects.\ \ref{sec:errM} and \ref{sec:URM}.

\section{Deviation measures for observables}\label{sec:measures}
Here we define the measures we use to quantify how well an observable $A'$ approximates a desired observable $A$. In this section we do not use the marginal condition \eqref{marginal}, so $A'$ is an arbitrary observable with the same outcome set $X$ as $A$, i.e., we drop all indices $i$ identifying the different observables. Our error quantities are operational in the sense that each is motivated by an experimental setup, which will in particular provide a natural way to measure them. All error definitions are based on the same cost function $c:X\times X\to\Rl$, where $c(x,y)$ is the ``cost'' of getting a result $x\in X$, when $y\in X$ would have been correct. The only assumptions are that $c(x,y)\geq0$ with $c(x,y)=0$ iff $x=y$.

As described above, we consider a quantum system with Hilbert space $\Cx^d$. As a reference observable $A$ we allow any complete von Neumann measurement on this system, that is, any observable whose the set $X$ of possible measurement outcomes has size $|X|=d$ and whose POVM elements $A(y)\in\cB(\Cx^d)$  ($y\in X$) are mutually orthogonal projectors of rank $1$; we can then also write $A(y)=\kettbra{\phi_y}$ with an orthonormal basis $\{\phi_y\}$ of $\Cx^d$. For the approximating observable $A'$ the POVM elements $A'(x)$ (with $x\in X$) are arbitrary with $A'(x)\geq0$  and $\sum_{x\in X}A(x)=\ii$.

The comparison will be based on a comparison of output distributions, for which we use the following notations: Given a quantum state $\rho$ on this system, i.e., a density operator with $\rho\geq0$ and $\tr\rho=1$, and an observable such as $A$, we will denote the outcome distribution by $\rho A$, so  $(\rho A)(y):=\tr(\rho A_y)$. This is a probability distribution on the outcome set $X$ and can be determined physically as the empirical outcome distribution after many experiments.

For comparing just two probability distributions $p:X\to\Rl_+$ and $q:X\to\Rl_+$, a canonical choice is the ``minimum transport cost''
\begin{equation}\label{c4probs}
  \cc(p,q):=\inf_\gamma\Bigl\{\sum_{xy}c(x,y)\gamma(x,y)\bigm| \gamma\ \mbox{couples}\ p\ \mbox{to}\ q\Bigr\},
\end{equation}
where the infimum runs over the set of all \emph{couplings}, or ``transport plans'' $\gamma:X\times X\to\Rl_+$ of $p$ to $q$, i.e., the set of all probability distributions $\gamma$ satisfying the marginal conditions $\sum_y\gamma(x,y)=p(x)$ and $\sum_{x}\gamma(x,y)=q(y)$. The motivations for this notion, and the methods to compute it efficiently are described in the Appendix. Since $X$ is finite, the infimum is over a compact set, so it is always attained. Moreover, since we assumed $c\geq0$ and $c(x,y)=0\Leftrightarrow x=y$, we also have $\cc(p,q)\geq0$ with equality iff $p=q$. If one of the distributions, say $q$, is concentrated on a point $\widetilde y$, only one coupling exists, namely $\gamma(x,y)=p(x)\delta_{y{\widetilde y}}$. In this case we abbreviate $\cc(p,q)=\cc(p,{\widetilde y})$, and get
\begin{equation}\label{ccx}
 \cc(p,{\widetilde y})=\sum_xp(x)c(x,{\widetilde y}),
\end{equation}
i.e., the average cost of moving all the points $x$ distributed according to $p$ to ${\widetilde y}$.

\subsection{Maximal measurement error $\errM{}$.}\label{sec:errM}
\begin{figure}[h]
\centering \def\svgwidth{0.7\textwidth} \input{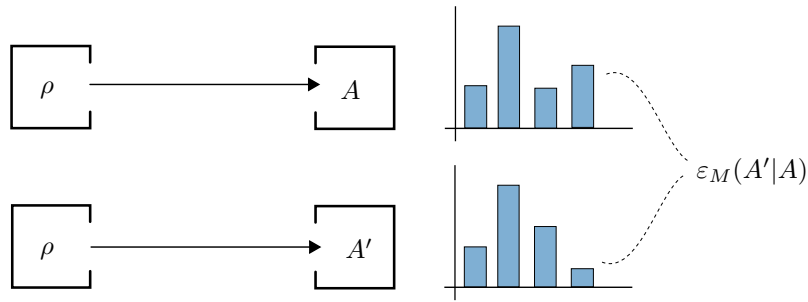}
\caption{For the maximal measurement error $\errM{}$ the transport distance of output distributions is maximized over all input states $\rho$.}
\end{figure}

The worst case error over \emph{all} input states is
\begin{equation}\label{define-errM}
\errM{}:=\sup_\rho\bigl\{\cc(\rho A',\rho A)\bigm|\rho\text{ quantum state on }\Cx^d\bigr\},
\end{equation}
which we call the \emph{maximal measurement error}.  Note that, like the cost function $c$ and the transport costs $\cc$, the measure $\errM{}$ need not be symmetric in its arguments, which is sensible as the reference and approximating observables have distinct roles. Similar definitions for the deviation of an approximating measurement from an ideal one have been made, for specific cost functions, in \cite{BLW1,BLW2} and \cite{AMU} before.

The definition (\ref{define-errM}) makes sense even if the reference observable $A$ is not a von Neumann measurement. Instead, the only requirement is that $A$ and $A'$ be general observables with the same (finite) outcome set $X$, not necessarily of size $d$. All our results below that involve only the maximal measurement error immediately generalize to this case as well.

One can see that it is expensive to determine the quantity $\errM{}$ experimentally according to the definition: one would have to measure and compare the outcome statistics $\rho A'$ and $\rho A$ for all possible input states $\rho$, which form a continuous set. The following definition of observable deviation alleviates this burden.

\subsection{Calibration error $\errC{}$.}\label{sect-define-errC}
\begin{figure}[h]
\centering \def\svgwidth{0.7\textwidth} \input{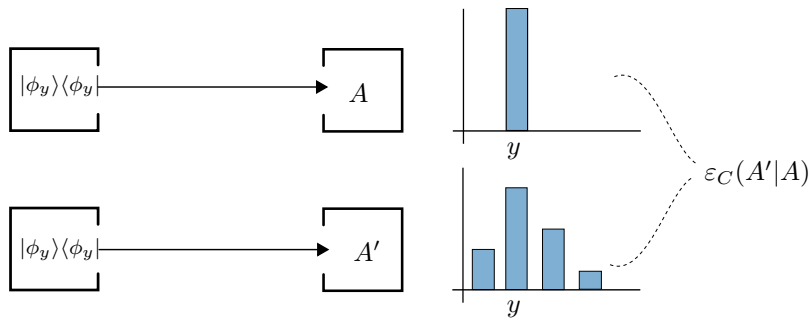}
\caption{For the calibration error $\errC{}$, the input state is constrained to the eigenstates of $A$, say with sharp $A$-value $y$, and the cost of moving the $A'$-distribution to $y$ is maximized over $y$.}
\end{figure}

Calibration is a process by which one tests a measuring device on inputs (or measured objects) for which the ``true value'' is known. Even in quantum mechanics we can set this up by demanding that the measurement of the reference observable on the input state gives a sharp value $y$. In a general scenario with continuous outcomes this can only be asked with a finite error $\delta$, which goes to zero at the end \cite{BLW1}, but in the present finite scenario we can just demand $(\rho A)(y)=1$. Since, for every outcome $y$ of a von Neumann measurement, there is only one state with this property (namely $\rho=\kettbra{\phi_y}$) we can simplify even further, and define the \emph{calibration error} by
\begin{equation}\label{define-errC}
\errC{}:=\sup_{y,\rho}\{\cc(\rho A',y)\bigm| \tr(\rho A(y))=1\}=\max_y \sum_x\brAAket{\phi_y}{A'(x)}{\phi_y}\ c(x,y).
\end{equation}

Note that the calibration idea only makes sense when there are sufficiently many states for which the reference observable has deterministic outcomes, i.e., for projective observables $A$.

A closely related quantity has recently been proposed by Appleby \cite{ApplebyNew}. It is formulated for real valued quantities with cost function $c(x,y)=(x-y)^2$, and has the virtue that it can be expressed entirely in terms of first and second moments of the probability distributions involved. So for any $\rho$, let $m$ and $v$ be the mean and variance of $\rho A$, and $v'$ the mean quadratic deviation of $\rho A'$ from $m$. Then Appleby defines
\begin{equation}\label{Appleby}
  \errA{}=\sup_\rho(\sqrt{v'}-\sqrt{v})^2.
\end{equation}
Here we added the square to make Appleby's quantity comparable to our variance-like (rather than standard deviation-like) quantities, and chose the letter $D$, because Appleby calls this the $D$-error. Since in the supremum we have also the states for which $A$ has a sharp distribution (i.e. $v=0$), we clearly have $\errA{}\geq\errC{}$. On the other hand, let $\Phi(x)=t(x-m)^2$ and $\Psi(y)=t/(1-t)(y-m)^2$ with some parameter $t\in(-\infty,1)$. Then one easily checks that $\Phi(x)-\Psi(y)\leq(x-y)^2$, so $(\Phi,\Psi)$ is a pricing scheme in the sense defined in the Appendix. Therefore
\begin{equation}\label{aappleest}
  \cc(\rho A',\rho A)\geq \sum_x(\rho A')(x)\Phi(x)-\sum_y(\rho A)(y)\Psi(y)= t\; v'-\frac t{1-t}\;v.
\end{equation}
Maximizing this expression over $t$ gives exactly \eqref{Appleby}. Therefore $\errC{}\leq\errA{}\leq\errM{}$.

\subsection{Entangled reference error $\errE{}$.}
\begin{figure}[h]
\centering \def\svgwidth{0.7\textwidth} \input{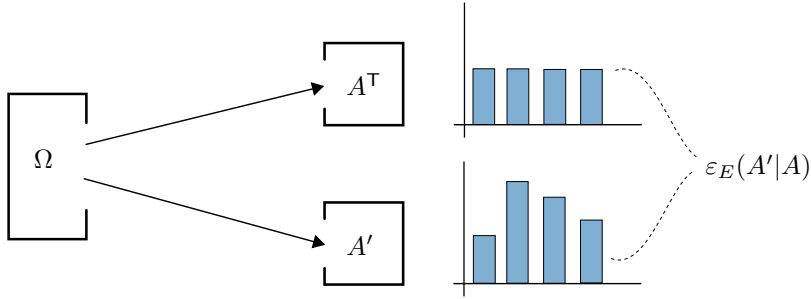}
\caption{The entangled reference error $\errE{}$ is a single expectation value, namely of the cost $c(x,y)$, where $y$ is the output of $A\trans$ and $x$ the output of $A'$. Like the other error quantities
   this expectation vanishes iff $A'=A$.}
\end{figure}
In quantum information theory a standard way of providing a reference state for later comparison is by applying a channel or observable to one half of a maximally entangled system. Two observables would be compared by measuring them (or suitable modifications) on the two parts of a maximally entangled system. Let us denote the entangled vector by $\Omega=d^{-1/2}\sum_k\ket{kk}$. Since later we will look at several distinct reference observables, the basis kets $\ket k$ in this expression have no special relation to $A$ or its eigenbasis $\phi_y$. We denote by $X\trans$ the transpose of an operator $X$ in the $\ket k$ basis, and by $A\trans$ the observable with POVM elements $A(y)\trans=\kettbra{\overline{\phi_y}}$, where $\overline{\phi_y}$ is the complex conjugate of $\phi_y$ in $\ket k$-basis. These transposes are needed due to the well-known relation $(X\otimes\idty)\Omega=(\idty\otimes X\trans)\Omega$. We now consider an experiment, in which $A'$ is measured on the first part and $A\trans$ on the second part of the entangled system, so we get the outcome pair $(x,y)$ with probability
\begin{equation}\label{entOut}
  p(x,y)=\brAAket\Omega{A'(x)\otimes A(y)\trans}\Omega=\brAAket\Omega{A'(x)A(y)\otimes\idty}\Omega=\frac1d \tr\bigl(A'(x)A(y)\bigr).
\end{equation}
As $A$ is a complete von Neumann measurement, this probability distribution is concentrated on the diagonal ($x=y$) iff $A'=A$, i.e., there are no errors of $A'$ relative to $A$. Averaging with the error costs we get a quantity we call the {\it entangled reference error}
\begin{equation}\label{define-errE}
\errE{}:=\sum_{xy} \frac1d \tr\bigl(A'(x)A(y)\bigr)\ c(x,y).
\end{equation}

Note that this quantity is measured as a single expectation value in the experiment with source $\Omega$. Moreover, when we later want to measure different such deviations for the various marginals, the source and the tested joint measurement device can be kept fixed, and only the various reference observables $A_i\trans$ acting on the second part need to be adapted suitably.

\subsection{Summary and comparison}

The quantities $\errM{}$, $\errC{}$ and $\errE{}$ constitute three different ways to quantify the deviation of an observable $A'$ from a projective reference observable $A$. Nevertheless, they are all based on the same distance-like measure, the cost function $c$ on the outcome set $X$. Therefore it makes sense to compare them quantitatively. Indeed, they are ordered as follows:
\begin{equation}
\errM{}\geq\errC{}\geq\errE{}.
\end{equation}
Here the first inequality follows by restricting the supremum \eqref{define-errM} to states which are sharp for $A$, and the second by noting the \eqref{define-errC} is the maximum of a function of $y$, of which \eqref{define-errE} is the average.

Moreover, as we argued before Eq.\ \eqref{define-errE}, $\errE{}=0$ if and only if $A=A'$, which is hence equivalent also to $\errM{}=0$ and $\errC{}=0$.

\section{Convexity of uncertainty diagrams}\label{sec:convex}
In this section we will consider tuples $(A_1,\ldots,A_n)$ of projection valued non-degenerated observables, as described in the introduction. We will collect some basic properties of the uncertainty regions $\UR L$, where $L\in\{M,C,E\}$, that is, 
\begin{equation}\label{URL}
  \UR L:=\Bigl\{\bigl(\errL 1,\ldots, \errL n\bigr)\Bigm| A_i'\ \mbox{marginals of a joint measurement}\Bigr\}.
\end{equation}

For two observables $B_1$ and $B_2$ with the same outcome set we can easily realize their mixture, or convex combination $B=tB_1+(1-t)B_2$ by flipping a coin with probability $t$ for heads in each instance and then apply $B_1$ when heads is up and $B_2$ otherwise. In terms of POVM elements this reads $B(x)=t B_1(x)+ (1-t) B_2(x)$.
We show first that this mixing operation does not increase the error quantities from Sect.~\ref{sec:measures}.

\begin{Lemma}
\label{convexity}
For $L\in\{M,D,C,E\}$ the error quantity $\errorL{B}$, is a convex function of $B$ , i.\,e. for $B=t B_1+(1-t)B_2$ and $t\in[0,1]$:
\begin{align}
\errorL{B}&\leq t\;\errorL{B_1}+(1-t)\;\errorL{B_1}.
\end{align}
\end{Lemma}

\begin{proof}
The basic fact used here is that the pointwise supremum of affine functions (i.e., those for which equality holds in the definition of a convex function) is convex. This is geometrically obvious, and easily verified from the definitions.
Hence we only have to check that each of the error quantities is indeed represented as a supremum of functions, which are affine in the observable $B$.

For $L=E$ we even get an affine function, because \eqref{define-errE} is linear in $A'$. For $L=C$ equation \eqref{define-errC} has the required form. For $L=M$ the definition \eqref{define-errM} is as a supremum, but the function $\cc$ is defined as an infimum. However, we can use the duality theory described in the Appendix (e.g. in \eqref{maxalpha}) to write it instead as a supremum over pricing schemes, of an expression which is just the expectation of $\Phi(x)$ plus a constant, and therefore an affine function. Finally, for Appleby's case \eqref{Appleby}, we get the same supremum, but over a subset of pricing schemes (the quadratic ones, see below \eqref{Appleby}).
\end{proof}

The convexity of the error quantities distinguishes measurement from preparation uncertainty. Indeed, the variances appearing in preparation uncertainty relations are typically concave functions, because they arise from minimizing the expectation of $(x-m)^2$ over $m$. Consequently, the preparation uncertainty regions may have gaps, and non-trivial behaviour on the side of large variances. The following proposition will show that measurement uncertainty regions are better behaved.

For every cost function $c$ on a set $X$ we can define a ``radius'' $\overline{c}^*$, the largest transportation cost from the uniform distribution (the ``center'' of the set of probability distributions) and a ``diameter''
$c^*$, the largest transportation cost between any two distributions:
\begin{equation}\label{radiusdia}
  \overline{c}^*=\max_y\sum_x c(x,y)/d \hskip 60pt c^*=\max_{xy}c(x,y).
\end{equation}

\begin{Proposition}\label{convexity-set}
Let $n$ observables $A_i$ and cost functions $c_i$ be given, and define $c^M_i=c^C_i=c^*_i$ and $c_i^E=\overline{c_i}^*$. Then, for $L\in\{M,C,E\}$, the uncertainty regions $\UR L$ is a convex set and has the following (monotonicity) property:
When $\vec x=(x_1,\ldots,x_n)\in\UR L$ and $\vec y=(y_i,\ldots,y_n)\in\Rl^n$ such that $x_i\leq y_i\leq c^L_i$,
then $\vec y\in\UR L$.
\end{Proposition}

\begin{proof}
Let us first clarify how to make the worst possible measurement $B$, according to the various error criteria, for which we go back to the setting of Sect.~\ref{sec:measures}, with just one observable $A$, and cost function $c$. In all cases, the worst measurement is one with constant and deterministic output, i.e., $B(x)=\delta_{x^*,x}\idty$. For $L=C$ and $L=M$ such a measurement will have $\errorL B=\max_y c(x^*,y)$, and we can choose $x^*$ to make this equal to $c^*=c^L$. For $L=E$ we get instead the average, which is maximized by $\overline c^*$.

We can now make a given joint measurement $R$ worse by replacing it partly by a bad one, say for the first observable $A_1$. That is, we set, for $\lambda\in[0,1]$,
\begin{equation}\label{Rbad}
  \widetilde R(x_1,x_2,\ldots,x_n)= \lambda B_1(x_1)\,\sum_{y_1}R(y_1,x_2,\ldots,x_n)+(1-\lambda)R(x_1,x_2,\ldots,x_n).
\end{equation}
Then all marginals $\widetilde A'_i$ for $i\neq1$ are unchanged, but $\widetilde A_1'(x_1)=\lambda B_1(x_1)+(1-\lambda)A'(x_1)$. Now as $\lambda$ changes from $0$ to $1$, the point in the uncertainty diagram will move continuously in the first coordinate direction from $\vec x$ to the point in which the first coordinate is replaced by its maximum value (see Fig.~\ref{fig:convexproof}(left)). Obviously, the same holds for every other coordinate direction, which proves the monotonicity statement of the proposition.
\begin{figure}[htb]
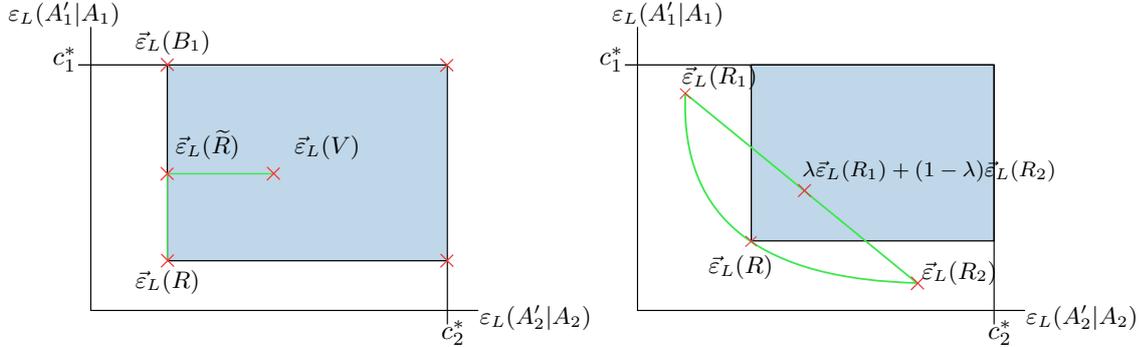


\centering \def\svgwidth{0.45\textwidth} \input{pics/convexproof.tex}
\centering \def\svgwidth{0.45\textwidth} \input{pics/convexproof2.tex} \caption{The blue shaded region corresponds to the monotonicity statement for $\vec\error_L(R)$. (left) $\widetilde R$  is a mixture of $R$ and $B_1$. We can also get an observable $V$ by mixing the second marginal of $\widetilde R$ with $B_2$ and thus reach every point in the blue shaded region. (right) $\vec\error_L(R)$ is componentwise convex. So the mixture of the points $\vec\error_L(R_1)$ and $\vec\error_L(R_2)$ is always in the monotonicity region corresponding to $\vec\error_L(R)$.}\label{fig:convexproof}
\end{figure} \\
Let $R_1$ and $R_2$ be two observables, and let $R=\lambda R_1+ (1-\lambda) R_2$ be their mixture. For proving the convexity of $\UR L$ we will have to show that every point on the line between  $\vec{\error}_L(R_1)$ and $\vec{\error}_L(R_2)$ can be attained by a tuple of errors corresponding to some allowed observable (see Fig.~\ref{fig:convexproof} (right)).
Now lemma \ref{convexity} tells us that every component of $\vec{\error}_L (R)$ is convex, which implies that $\vec{\error}_L(R)\leq \lambda\vec{\error}_L(R_1)+(1-\lambda)\vec{\error}_L(R_2)$. But, by monotonicity, this also means that $\lambda\vec{\error}_L(R_1)+(1-\lambda)\vec{\error}(R_2)$ is in $\mathcal{U}_L$ again, which shows the convexity of $\mathcal{U}_L$.
\end{proof}

\subsection{Example: Phase space pairs}\label{sec:PhasSpas}
As is plainly visible from Fig.~\ref{fig:URs}, the three error criteria considered here usually give different results. However, under suitable circumstances they all coincide. This is the case for conjugate pairs related by Fourier transform \cite{WerFou}. The techniques needed to show this are the same as for the standard position/momentum case \cite{BLW2,QHA}, and in addition imply that the region for preparation uncertainty is also the same.

In the finite case there is not much to choose: We have to start from a finite abelian group, which we think of as position space, and its dual group, which is then the analogue of momentum space. The unitary connecting the two observables is the finite Fourier associated with the group. The cost function needs to be translation invariant, i.e., $c(x,y)=c(x-y)$. Then, by an averaging argument, we find for all error measures that a covariant phase space observable minimizes measurement uncertainty (all three versions). The marginals of such an observable can be simulated by first doing the corresponding reference measurement, and then adding some random noise. This implies \cite{AMU} that $\errM{}=\errC{}$. But we know more about this noise: It is independent of the input state so that the average and the maximum of the noise (as a function of the input) coincide, i.e.,  $\errC{}=\errE{}$. Finally, we know that the noise of the position marginal is distributed according to the position distribution of a certain quantum state which is, up to normalization and a unitary parity inversion, the POVM element of the covariant phase space observable at the origin. The same holds for the momentum noise. But then the two noise quantities are exactly related like the position and momentum distributions of a state, and the tradeoff curve for that problem is exactly preparation uncertainty, with variance criteria based on the same cost function.

\begin{figure}[H]
\centering
\def\svgwidth{0.5\textwidth} \input{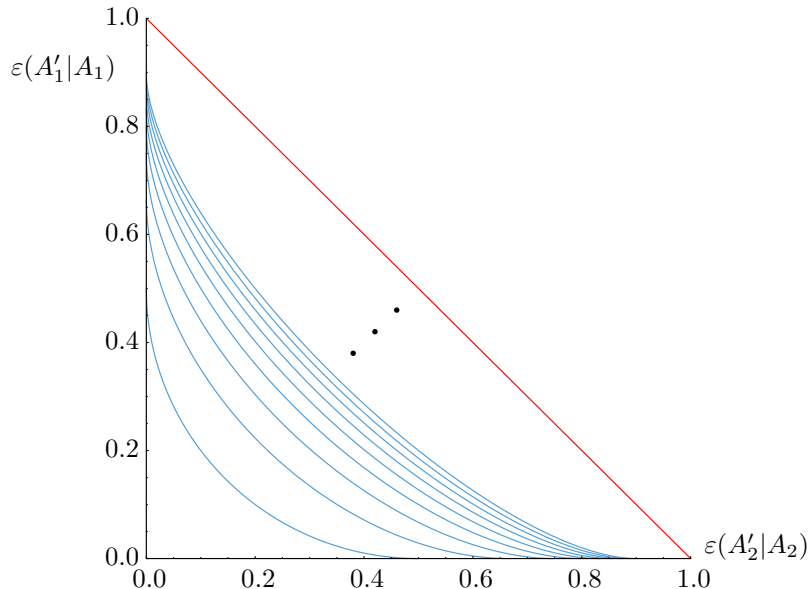}\caption{The uncertainty tradeoff curves for discrete position/momentum pairs, with discrete metric. In this case all uncertainty regions, also the one for preparation uncertainty, coincide. The parameter of the above tradeoff curves is the order $d=2,3,\ldots,10,\cdots,\infty$ of the underlying abelian group. \label{fig:phsp}}
\end{figure}

If we choose the discrete metric for $c$, the uncertainty region depends only on the number $d$ of elements in the group we started from \cite{WerFou}. The largest $\error$ for all quantities is the distance from a maximally mixed state to any pure state, which is $\Delta=(1-1/d)$. The exact tradeoff curve is then an ellipse, touching the axes at the points $(0,\Delta)$ and $(\Delta,0)$. The resulting family of curves, parameterized by $d$, is shown in Fig.~\ref{fig:phsp}. In general, however, the tradeoff curve requires the solution of a non-trivial family of ground state problems, and cannot be given in closed form. For bit strings of length $n$, and the cost some convex function of Hamming distance there is an expression for large $n$ \cite{WerFou}.
\section{Computing uncertainty regions via semidefinite programming}\label{sec:sdp}
We show here how the uncertainty regions -- and therefore optimal uncertainty relations -- corresponding to each of the three error measures can actually be computed, for any given set of projective observables $A_1,\ldots,A_n$ and cost functions $c_1,\ldots,c_n$. Our algorithms will come in the form of semidefinite programs (SDPs) \cite{boyd-vandenberghe-book,boyd-vandenberghe}, facilitating efficient numerical computation of the uncertainty regions via the many existing program packages to solve SDPs. Moreover, the accuracy of such numerical results can be rigorously certified via the duality theory of SDPs. To obtain the illustrations in this paper we used the CVX package \cite{cvx,gb08} under MATLAB.

As all our uncertainty regions $\UR L\subset\Rl^n$ (for $L=M,C,E$) are convex and closed (Sect.~\ref{sec:convex}), they are completely characterized by their supporting hyperplanes (for a reference to convex geometry see \cite{rockafellar}). Due to the monotonicity property stated in Prop.~\ref{convexity-set} some of these hyperplanes  just cut off the set parallel along the planes $x_i=c^L_i$. The only hyperplanes of interest are thus those with nonnegative normal vectors $\vec w=(w_1,\ldots,w_n)\in\Rl_+^n$ (see Fig.~\ref{fig:sdp1}). Each hyperplane is completely specified by its ``offset'' $b_L(\vec w)$ away from the origin, and this function determines $\UR L$:
\begin{eqnarray}\label{define-offset}
b_L(\vec w)&:=&\inf\Bigl\{\vec w\cdot\vec\error\Bigm| \vec\error\in\UR L\Bigr\}\ , \\
\UR L&=&\Bigl\{\vec\error\in\Rl^n\,\Bigm|\,\forall\vec w\in\Rl_+^n:\
                      \vec w\cdot\vec\error\geq b_L(\vec w)\ \mbox{and }\forall i:\error_i\leq c^L_i \Bigr\}.
\end{eqnarray}
In fact, due to homogeneity $b_L(t\vec w)=t\,b_L(\vec w)$ we can restrict everywhere to the subset of vectors $\vec w\in\Rl_+^n$ that, for example, satisfy $\sum_iw_i=1$, suggesting an interpretation of the $w_i$ as weights of the different uncertainties $\error_i$. Our algorithms will, besides evaluating $b_L(\vec w)$, also allow to compute an (approximate) minimizer $\vec\error$, so that one can plot the boundary of the uncertainty region $\UR L$ by sampling over $\vec w$, which is how the figures in this paper were obtained.
\begin{figure}[htb]
\centering \def\svgwidth{0.5\textwidth} \input{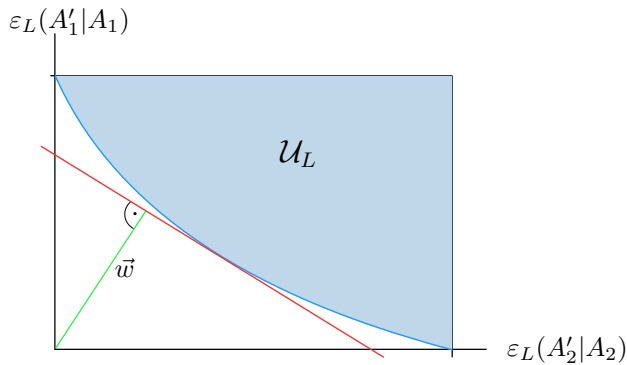}
\caption{The lower bound of the uncertainty region $\UR{L}$ can be described by its supporting hyperplanes (red line) with a normal vector $\vec{w}\in\mathbb{R}_+^n$
.}
\label{fig:sdp1}
\end{figure}

Let us further note that knowledge of $b_L(\vec w)$ for some $\vec w\in\Rl_+^n$ immediately yields a quantitative uncertainty relation: every error tuple $\vec\error\in\UR L$ attainable via a joint measurement is constrained by the affine inequality $\vec w\cdot\vec\error\geq b_L(\vec w)$, meaning that some weighted average of the attainable error quantities $\error_i$ cannot become too small. When $b_L(\vec w)>0$ is strictly positive, this excludes in particular the zero error point $\vec\error=\vec0$. The obtained uncertainty relations are \emph{optimal} in the sense that there exists $\vec\error\in\UR L$ which attains strict equality $\vec w\cdot\vec\error=b_L(\vec w)$.

Having reduced the computation of an uncertainty region essentially to determining $b_L(\vec w)$ (possibly along with an optimizer $\vec\error$), we now treat each case $L=M,C,E$ in turn.

\subsection{Computing the uncertainty region $\UR M$}\label{sec:URM}
On the face of it, the computation of the offset $b_M(\vec w)$ looks daunting: expanding the definitions we obtain
\begin{align}
b_M(\vec w)=\inf_R\,\sum_{i=1}^nw_i\sup_{\rho}\,\check c_i(\rho A'_i,\rho A_i),
\end{align}
where the infimum runs over all joint measurements $R$ with outcome set $X_1\times\ldots\times X_n$, inducing the marginal observables $A'_i=A'_i(R)$ according to (\ref{marginal}), and the supremum over all sets of $n$ quantum states $\rho_1,\ldots,\rho_n$, and where the transport costs $\cc_i(p,q)$ are given as a further infimum \eqref{c4probs} over the couplings $\gamma_i$ of $p=\rho A'_i$ and $q=\rho A_i$.

The first simplification is to replace the infimum over each coupling $\gamma_i$, via a dual representation of the transport costs, by a maximum over \emph{optimal pricing schemes} $(\Phi_\alpha,\Psi_\alpha)$, which are certain pairs of functions $\Phi_\alpha,\Psi_\alpha:X_i\to\Rl$, where $\alpha$ runs over some finite label set $\ops_i$. The characterization and computation of the pairs $(\Phi_\alpha,\Psi_\alpha)$, which depend only on the chosen cost function $c_i$ on $X_i$,  is described in the Appendix. The simplified expression for the optimal transport costs is then
\begin{equation}
\cc_i(p,q)=\max_{\alpha\in\ops_i}\sum_x\Phi_{\alpha}(x)\, p(x)-\sum_y\Psi_{\alpha}(y)q(y).
\end{equation}
We can then continue our computation of $b_M(\vec w)$:
\begin{align}
b_M(\vec w)&=\inf_R\,\sum_iw_i\sup_{\rho}\max_{\alpha\in\ops_i}\Bigl(\sum_x\Phi_\alpha(x)\tr[\rho A'_i(x)]-\sum_y\Psi_\alpha(y)\tr[\rho A_i(y)]\Bigr)\\
&=\inf_R\sum_iw_i\max_{\alpha\in\ops_i}\sup_{\rho}\tr\Bigl[\rho\Bigl(\sum_x\Phi_\alpha(x)A'_i(x)-\sum_y\Psi_\alpha(y)A_i(y)\Bigr)\Bigr]\\
&=\inf_R\sum_iw_i\max_{\alpha\in\ops_i}\maxEV\Bigl(\sum_x\Phi_\alpha(x)A'_i(x)-\sum_y\Psi_\alpha(y)A_i(y)\Bigr),
\end{align}
where $\maxEV(B_{i,\alpha})$ denotes the maximum eigenvalue of a Hermitian operator $B_{i,\alpha}$. Note that $\maxEV(B_{i,\alpha})=\inf\{\mu_i\,|\,B_{i,\alpha}\leq\mu_i\ii\}$, which one can also recognize as the dual formulation of the convex optimization $\sup_\rho\tr(\rho B_{i,\alpha})$ over density matrices, so that
\begin{equation}\label{maxxB}
  \max_{\alpha\in\ops_i}\maxEV(B_{i,\alpha})=\inf\{\mu_i\,|\,\forall\alpha\in\ops_i:\,B_{i,\alpha}\leq\mu_i\ii\}
\end{equation}
We obtain thus a single constrained minimization:
\begin{align}
b_M(\vec w)=\inf_{R,\{\mu_i\}}\Bigl\{\sum_iw_i\mu_i\bigm|\forall i\forall\alpha\in\ops_i:\,\sum_x\Phi_\alpha(x)A'_i(x)-\sum_y\Psi_\alpha(y)A_i(y)\leq\mu_i\ii\Bigr\}.
\end{align}
Making the constraints on the POVM elements $R(x_1,\ldots,x_n)$ of the joint observable $R$ explicit and expressing the maginal observables $A'_i=A'_i(R)$ directly in terms of them by (\ref{marginal}), we finally obtain the following SDP representation for the quantity $b_M(\vec w)$:
\begin{equation}\label{primalSDPM}\boxformula{\begin{array}{rl}
&b_M(\vec w)=\inf~\sum_iw_i\mu_i \\
&\text{with real variables\ }\mu_i\ \text{and $d\times d$-matrix variables\ }R(x_1,\ldots,x_n)\ \text{subject to }\\
           &\begin{array}{rl} \mu_i\ii&\geq\sum_{x_1,\ldots,x_n}\Phi_{\alpha}(x_i)\,R(x_1,\ldots,x_n)-\sum_y\Psi_{\alpha}(y)A(y)\quad\forall i\,\forall \alpha\in\ops_i\\
                   R(x_1,\ldots,x_n)&\geq0\quad\forall x_1,\ldots,x_n\\
                   \sum_{x_1,\ldots,x_n}R(x_1,\ldots,x_n)&=\ii.\end{array}\end{array}}
\end{equation}
The derivation above shows further that, when $w_i>0$, the $\mu_i$ attaining the infimum equals $\mu_i=\sup_\rho\cc_i(\rho A'_i,\rho A_i)=\errM{i}$, where $A'_i$ is the marginal coming from a corresponding optimal joint measurement $R(x_i,\ldots,x_n)$. Since numerical SDP solvers usually output an (approximate) optimal variable assignment, one obtains in this way directly a boundary point $\vec\error=(\mu_1,\ldots,\mu_n)$ of $\UR M$ when all $w_i$ are strictly positive. If $w_i=0$ vanishes, a corresponding boundary point $\vec\error$ can be computed via $\error_i=\errM{i}=\max_{\alpha\in\ops_i}\maxEV(\sum_{x_1,\ldots,x_n}\Phi_{\alpha}(x_i)\,R(x_1,\ldots,x_n)-\sum_y\Psi_{\alpha}(y)A(y))$ from an optimal assignment for the POVM elements $R(x_1,\ldots,x_n)$.

For completeness we also display the corresponding dual program \cite{boyd-vandenberghe-book} (note that strong duality holds, and the optima of both the primal and the dual problem are attained):
\begin{align}\label{dualSDPM}\boxformula{\begin{array}{rl}
&b_M(\vec w)=\sup~\tr[C]-\sum_{i,\alpha}\tr[D_{i,\alpha}\sum_y\Psi_\alpha(y)A_i(y)] \\
&\text{with $d\times d$-matrix variables\ }C\ \text{and}\ D_{i,\alpha}\ \text{subject to }\\
           &\begin{array}{rl} C&\leq\sum_{i,\alpha}\Phi_{\alpha}(x_i)D_{i,\alpha}\quad\forall x_1,\ldots,x_n\\
                   0&\leq D_{i,\alpha}\quad\forall i\,\forall \alpha\in\ops_i\\
                   w_i&=\sum_\alpha\tr[D_{i,\alpha}]\quad\forall i.\end{array}\end{array}}
\end{align}

\subsection{Computing the uncertainty region $\UR C$}
To compute the offset function $b_C(\vec w)$ for the calibration uncertainty region $\UR C$ we use the last form in (\ref{define-errC}) and recall that the projectors onto the sharp eigenstates of $A_i$ (see Sect.\ \ref{sect-define-errC}) are exactly the POVM elements $A_i(x)$ for $x\in X_i$:
\begin{align}
b_C(\vec w)&=\inf_R~\sum_iw_i\max_y\,\sum_x\tr[A'_i(x)A_i(y)]c_i(x,y)\label{definebCw}\\
&=\inf_R~\sum_iw_i\sup_{\{\lambda_{i,y}\}}\,\sum_y\lambda_{i,y}\sum_x\tr[A'_i(x)A_i(y)]c_i(x,y)\\
&=\inf_R\sup_{\{\lambda_{i,y}\}}\sum_{x_1,\ldots,x_n}\tr\Bigl[R(x_1,\ldots,x_n)\sum_{i,y}w_i\lambda_{i,y}c_i(x_i,y)A_i(y)\Bigr]\label{beforeminimaxthm}
\end{align}
where again the infimum runs over all joint measurements $R$, inducing the marginals $A'_i$, and we have turned, for each $i=1,\ldots,n$, the maximum over $y$ into a linear optimization over probabilities $\lambda_{i,y}\geq0$ ($y=1,\ldots,d$) subject to the normalization constraint $\sum_y\lambda_{i,y}=1$. In the last step, we have made the $A'_i$ explicit via (\ref{marginal}).

The first main step towards a tractable form is von Neumann's minimax theorem \cite{Nikaido,Sion}:
As the sets of joint measurements $R$ and of probabilities $\{\lambda_{i,y}\}$ are both convex and the optimization function is an affine function of $R$ and, separately, also an affine function of the $\{\lambda_{i,y}\}$, we can interchange the infimum and the supremum:
\begin{align}
b_C(\vec w)=\sup_{\{\lambda_{i,y}\}}\inf_R\sum_{x_1,\ldots,x_n}\tr\Bigl[R(x_1,\ldots,x_n)\sum_{i,y}w_i\lambda_{i,y}c_i(x_i,y)A_i(y)\Bigr].
\end{align}
The second main step is to use SDP duality \cite{boyd-vandenberghe} to turn the constrained infimum over $R$ into a supremum, abbreviating the POVM elements as $R(x_1,\ldots,x_n)=R_\xi$:
\begin{align}
\inf_{\{R_\xi\}}\Bigl\{\sum_\xi R_\xi B_\xi\Bigm| R_\xi\geq0~\forall\xi,~\sum_\xi R_\xi=\ii\Bigr\}~=~\sup_Y\Bigl\{\tr[Y]\Bigm|Y\leq B_\xi~\forall\xi\Bigr\},
\end{align}
which is very similar to a dual formulation often employed in optimal ambiguous state discrimination \cite{Holevo,Yuen}.

Putting everything together, we arrive at the following SDP representation for the offset quantity $b_C(\vec w)$:
\begin{equation}\label{primalSDPC}\boxformula{\begin{array}{rl}
&b_C(\vec w)=\sup~\tr[Y] \\
&\text{with real variables\ }\lambda_{i,y}\ \text{and a $d\times d$-matrix variable\ }Y\ \text{subject to }\\
           &\begin{array}{rl} Y&\leq\sum_{i,y}w_i\lambda_{i,y}c_i(x_i,y)A_i(y)\quad\forall x_1,\ldots,x_n\\
                   \lambda_{i,y}&\geq0\quad\forall i\,\forall y\\
                   \sum_y\lambda_{i,y}&=1\quad\forall i.\end{array}\end{array}}
\end{equation}
The dual SDP program reads (again, strong duality holds, and both optima are attained):
\begin{equation}\label{dualSDPC}\boxformula{\begin{array}{rl}
&b_C(\vec w)=\inf~\sum_iw_im_i \\
&\text{with real variables\ }m_i\ \text{and $d\times d$-matrix variables\ }R(x_1,\ldots,x_n)\ \text{subject to }\\
           &\begin{array}{rl} m_i&\geq\sum_{x_1,\ldots,x_n}\tr\bigl[R(x_1,\ldots,x_n)A_i(y)\bigl]c_i(x_i,y)\quad\forall i\,\forall y\\
                   R(x_1,\ldots,x_n)&\geq0\quad\forall x_1,\ldots,x_n\\
                   \sum_{x_1,\ldots,x_n}R(x_1,\ldots,x_n)&=\ii.\end{array}\end{array}}
\end{equation}
This dual version can immediately be recognized as a translation of Eq.\ (\ref{definebCw}) into SDP form, via an alternative way of expressing the maximum over $y$ (or via the linear programming dual of $\sup_{\{\lambda_{i,y}\}}$ from Eq.\ (\ref{beforeminimaxthm})).

To compute a boundary point $\vec\error$ of $\UR C$ lying on the supporting hyperplane with normal vector $\vec w$, it is best to solve the dual SDP (\ref{dualSDPC}) and obtain $\vec\error=(m_1,\ldots,m_n)$ from an (approximate) optimal assignment of the $m_i$. Again, this works when $w_i>0$, whereas otherwise one can compute $\error_i=\max_y\sum_{x_1,\ldots,x_n}\tr[R(x_1,\ldots,x_n)A_i(y)]c_i(x_i,y)$ from an optimal assingment of the $R(x_1,\ldots,x_n)$. From many primal-dual numerical SDP solvers (such as CVX \cite{cvx,gb08}), one can alternatively obtain optimal POVM elements $R(x_1,\ldots,x_n)$ also from solving the primal SDP (\ref{primalSDPC}) as optimal dual variables corresponding to the constraints $Y\leq\ldots$, and compute $\vec\error$ from there.

\subsection{Computing the uncertainty region $\UR E$}
As one can see by comparing the last expressions in the defining equations (\ref{define-errC}) and (\ref{define-errE}), respectively, the evaluation of $b_E(\vec w)$ is quite similar to (\ref{definebCw}), except that the maximum over $y$ is replaced by a uniform average over $y$. This simply corresponds to fixing $\lambda_{i,y}=1/d$ for all $i,y$ in Eq.\ (\ref{beforeminimaxthm}), instead of taking the supremum. Therefore, the primal and dual SDPs for the offset $b_E(\vec w)$ are
\begin{equation}\label{primalSDPE}\boxformula{\begin{array}{rl}
&b_E(\vec w)=\sup~\frac{1}{d}\tr[Y] \\
&\text{with a $d\times d$-matrix variable\ }Y\ \text{subject to }\\
           &\begin{array}{rl} Y&\leq\sum_{i,y}w_ic_i(x_i,y)A_i(y)\quad\forall x_1,\ldots,x_n.\end{array}\end{array}}
\end{equation}
and
\begin{equation}\label{dualSDPE}\boxformula{\begin{array}{rl}
&b_E(\vec w)=\inf~\frac{1}{d}\sum_i\sum_y\sum_{x_1,\ldots,x_n}w_i\tr\bigl[R(x_1,\ldots,x_n)A_i(y)\bigl]c_i(x_i,y) \\
&\text{with $d\times d$-matrix variables\ }R(x_1,\ldots,x_n)\ \text{subject to }\\
           &\begin{array}{rl} R(x_1,\ldots,x_n)&\geq0\quad\forall x_1,\ldots,x_n\\
                   \sum_{x_1,\ldots,x_n}R(x_1,\ldots,x_n)&=\ii.\end{array}\end{array}}
\end{equation}
The computation of a corresponding boundary point $\vec\error\in\UR E$ is similar as above.

\section*{Acknowledgements}
\noindent We thank Oliver Sachse and Kais Abdelkhalek for valuable discussions. The authors acknowledge financial support from the BMBF project Q.com-Q, the DFG project WE1240/20 and the ERC grant
DQSIM.

\renewcommand{\thesubsection}{A.\arabic{subsection}}
\setcounter{subsection}{0}

\section*{Appendix: Optimal Transport}\label{sec:transport}
\subsection{Kantorovich duality}

In this appendix we collect the basic theory of optimal transport adapted to the finite setting at hand. This eliminates all the topological and measure theoretic fine points that can be found, e.g., in Villani's book \cite{Villani}, which we also recommend for extended proofs of the statements in our summary. We slightly generalize the setting from the cost functions used in the main text of this paper: We allow the two variables on which the cost function depends to range over different sets. This might actually be useful for comparing observables, which then need not have the same outcome sets. Which outcomes are considered to be close or the same must be specified in terms of the cost function. We introduce this generalization here less for the sake of applications rather than for a simplification of the proofs, in particular for the book-keeping of paths in the proof of Lemma~\ref{lem:ccm=E}.

The basic setting is that of two finite sets $X$ and $Y$, and a arbitrary function $c:X\times Y\to\Rl$, called the \emph{cost function}.
The task is to optimize the transport of some distribution of stuff on $X$, described by a distribution function $p:X\to\Rl_+$, to a final distribution $q:Y\to\Rl_+$ on $Y$ when the transportation of one unit of stuff from the point $x$ to the point $y$ costs $c(x,y)$. In the first such scenario ever considered, namely by Gaspar Monge, the ``stuff'' was earth, the distribution $p$ a hill, and $q$ a fortress. Villani \cite{Villani} likes to phrase the scenario in terms of bread produced at bakeries $x\in X$ to be delivered to caf\'es $y\in Y$. This makes plain that optimal transport is sometimes considered a branch of mathematical economics, and indeed Leonid Kantorovich, who created much of the theory, received a Nobel prize in economics. In our case the ``stuff'' will be probability.

A {\it transport plan} (or \emph{coupling}) will be a probability distribution $\gamma:X\times Y\to\Rl_+$, which encodes how much stuff is moved from any $x$ to any $y$. Since all of $p$ is to be moved, $\sum_y\gamma(x,y)=p(x)$, and since all stuff is to be delivered, $\sum_x\gamma(x,y)=q(y)$. Now, for any transport plan $\gamma$ we get a total cost of $\sum_{x,y}\gamma(x,y)c(x,y)$, and we are interested in the optimum
\begin{equation}\label{primal}
  \check c(p,q)=\inf_\gamma\Bigl\{\sum_{xy}c(x,y)\gamma(x,y)\bigm| \gamma\ \mbox{couples}\ p\ \mbox{to}\ q\Bigr\}.
\end{equation}
This is called the primal problem, to which there is also a dual problem. In economic language it concerns {\it pricing schemes}, that is, pairs of functions $\Phi:X\to\Rl$ and $\Psi:Y\to\Rl$ satisfying the inequality
\begin{equation}\label{pricing}
  \Phi(x)-\Psi(y)\leq c(x,y) \quad \mbox{for all }\ x\in X,\ y\in Y,
\end{equation}
and demands to maximize
\begin{equation}\label{dual}
  \hat c(p,q)=\sup_{\Phi,\Psi}\Bigl\{\sum_{x}\Phi(x)p(x)-\sum_{y}\Psi(y)q(y)\bigm| (\Phi,\Psi)\ \mbox{is a pricing scheme}\Bigr\}.
\end{equation}

In Villani's example \cite{Villani}, think of a consortium of bakeries and caf\'es, that used to organize the transport themselves according to some plan $\gamma$. Now they are thinking of hiring a contractor, which offers to do the job, charging $\Phi(x)$ for every unit picked up from bakery $x$, and giving $\Psi(y)$ to caf\'e $y$ on delivery (these numbers can be negative). Their offer is that this will reduce overall costs, since their pricing scheme satisfies \eqref{pricing}. Indeed, the overall charge to the consortium will be
\begin{equation}\label{trivial}
  \sum_{x}\Phi(x)p(x)-\sum_{y}\Psi(y)q(y)=\sum_{xy}\bigl(\Phi(x)-\Psi(y)\bigr)\gamma(x,y)\leq\sum_{xy}c(x,y)\gamma(x,y).
\end{equation}
Taking the $\sup$ on the left hand side of this inequality (the company will try to maximize their profits by adjusting the pricing scheme $(\Phi,\Psi)$) and the $\inf$ on the right hand side (the transport plan $\gamma$ was already optimized), we get
$\hat c(p,q)\leq\check c(p,q)$. It can be shown via the general duality theory of linear programming \cite{boyd-vandenberghe-book} that the duality gap closes in this case, i.e., we actually always have
\begin{equation}\label{duality}
  \hat c(p,q)=\check c(p,q).
\end{equation}
So the consortium will face the same transport costs in the end if the contractor chooses an optimal pricing scheme. (Note that both the infimum and the supremum in the definitions of $\check c$ and $\hat c$, respectively, are attained as $X$ and $Y$ are finite sets.)

What is especially interesting for us, however, is that the structure of the optimal solutions for both variational problems is very special, and  both problems can be reduced to a combinatorial optimization over finitely many possibilities, which furthermore can be constructed independently of $p$ and $q$. Indeed, pricing schemes and transport plans are both related to certain subsets of $X\times Y$. We define $S(\gamma)\subseteq X\times Y$ as the {\it support} of $\gamma$, i.e., the set of pairs on which $\gamma(x,y)>0$. For a pricing scheme $(\Phi,\Psi)$ we define the {\it equality set} $E(\Phi,\Psi)$ as the set of points $(x,y)$ for which equality holds in \eqref{pricing}. Then equality holds in \eqref{trivial} if and only if $S(\gamma)\subset E(\Phi,\Psi)$. Note that for $\gamma$ to satisfy the marginal condition for given $p$ and $q$, its support $S(\gamma)$ cannot become too small (depending on $p$ and $q$). On the other hand, $E(\Phi,\Psi)$ cannot be too large, because the resulting system of equations for $\Phi(x)$ and $\Psi(y)$ would become overdetermined and inconsistent. The kind of set for which they meet is described in the following Definition.

\begin{Definition}\label{def:ccm}Let $X,Y$ be finite sets and $c:X\times Y\to\Rl$ a function. Then a subset $\Gamma\subset X\times Y$ is called {\bf cyclically $c$-monotone} (``ccm'' for short), if for any sequence of distinct pairs $(x_1,y_1)\in\Gamma, \ldots,(x_n,y_n)\in\Gamma$, and any permutation $\pi$ of $\{1,\ldots,n\}$ the inequality
\begin{equation}\label{ccm}
  \sum_{i=1}^n c(x_i,y_i)\leq \sum_{i=1}^n c(x_i,y_{\pi i})
\end{equation}
holds. When $\Gamma$ is not properly contained in another cyclically $c$-monotone set, it is called {\bf maximally} cyclically $c$-monotone   (``mccm'' for short).
\end{Definition}

A basic example of a ccm set is the equality set $E(\Phi,\Psi)$ for any pricing scheme $(\Phi,\Psi)$. Indeed, for $(x_i,y_i)\in E(\Phi,\Psi)$ and any permutation $\pi$ we have
\begin{equation}\label{E=ccm}
  \sum_{i=1}^n c(x_i,y_i)=\sum_{i=1}^n \bigl(\Phi(x_i)-\Psi(y_i)\bigr))=\sum_{i=1}^n \bigl(\Phi(x_i)-\Psi(y_{\pi i})\bigr)
  \leq \sum_{i=1}^n c(x_i,y_{\pi i})
\end{equation}
The role of ccm sets in the variational problems (\ref{primal}) and (\ref{dual}) is summarized in the following proposition.

\begin{Proposition}\label{prop:extremeTransport}Let $X,Y,c,p,q$ be given as above. Then
\begin{itemize}
\item[(1)] A coupling $\gamma$ minimizes \eqref{primal} if and only if $S(\gamma)$ is ccm.
\item[(2)] The dual problem  \eqref{dual} has a maximizer $(\Phi,\Psi)$ for which $E(\Phi,\Psi)$ is mccm.
\item[(3)] If $\Gamma\subseteq X\times Y$ is mccm, there is a pricing scheme $(\Phi,\Psi)$ with $E(\Phi,\Psi)=\Gamma$, and $(\Phi,\Psi)$ is uniquely determined by $\Gamma$ up to the addition of the same constant to $\Phi$ and to $\Psi$.
\end{itemize}
\end{Proposition}

\begin{proof}[Sketch of proof]\strut\\
(1) Suppose  $(x_i,y_i)\in S(\gamma)$ ($i=1,\ldots,n$), and let $\pi$ be any permutation. Set $\delta=\min_i\gamma(x_i,y_i)$. Then we can modify $\gamma$  by subtracting $\delta$ from any $\gamma(x_i,y_i)$ and adding $\delta$ to $\gamma(x_i,y_{\pi i})$. This operation keeps $\gamma\geq0$ and does not change the marginals. The target functional in the infimum (\ref{primal}) is changed by $\delta$ times the difference of the two sides of \eqref{ccm}. For a minimizer $\gamma$ this change must be $\geq0$, which gives inequality \eqref{ccm}. For the converse we need a Lemma, whose proof will be sketched below.

\begin{Lemma}\label{lem:ccm=E}
For any ccm set $\Gamma$ there is some pricing scheme $(\Phi,\Psi)$ with $E(\Phi,\Psi)\supseteq\Gamma$.
\end{Lemma}

\noindent By applying this to $\Gamma=S(\gamma)$ we find that the duality gap closes for $\gamma$, i.e., equality holds in \eqref{trivial}, and hence $\gamma$ is a minimizer.

(2) Every subset $\Gamma\subset X\times Y$ can be thought of as a bipartite graph with vertices $X\cup Y$ and an edge joining $x\in X$ and $y\in Y$ iff $(x,y)\in\Gamma$ (see Fig.~\ref{fig:graph}). We call $\Gamma$ connected, if any two vertices are linked by a sequence of edges. Consider now the equality set $E(\Phi,\Psi)$ of some pricing scheme. We modify $(\Phi,\Psi)$ by picking some connected component, and setting  $\Phi'(x)=\Phi(x)+a$ and $\Psi'(y)=\Psi(y)+a$ for all $x,y$ in that component.  If $\abs a$ is sufficiently small, $(\Phi',\Psi')$ will still satisfy all the inequalities \eqref{pricing}, and $E(\Phi',\Psi')=E(\Phi,\Psi)$. The target functional in the optimization (\ref{dual}) depends linearly on $a$, so moving in the appropriate direction will increase, or at least not decrease it. We can continue until another one of the inequalities \eqref{pricing} becomes tight. At this point $E(\Phi',\Psi')\supsetneq E(\Phi,\Psi)$. This process can be continued until the equality set $E(\Phi,\Psi)$ is connected. Then $(\Phi,\Psi)$ is uniquely determined by $E(\Phi,\Psi)$ up to a common constant.

\begin{figure}[htb]
\centering
\def\xs{7} \def\ys{9} 
\definecolor{blu}{rgb}{0.49,0.68,0.82}    
\def\lk#1#2{\draw[fill=blu] (#1+.05,#2+.05) rectangle +(.9,.9);
            \draw[thick] (\xs,1.2*#1) node[circle,draw,fill=white] {#1}-- (\ys,1.2*#2)
                                       node[circle,draw,fill=white] {#2};   }
\def\lab#1{\node at (.5+#1,.5) {#1}; \node at (.5,.5+#1) {#1};}
\begin{tikzpicture}
  \draw (1,1) grid (5,5) ;
  \lk11 \lk13 \lk21 \lk32 \lk33 \lk34 \lk43
  \lab1 \lab2 \lab3 \lab4
  \node at (3,-.2) {$X$} ;   \node at (-.2,3) {$Y$} ;
  \node at (\xs,0) {$X$} ;   \node at (\ys,0) {$Y$} ;
\end{tikzpicture}
\caption{Representation of a subset $\Gamma\subset X\times Y$ (left) as a bipartite graph (right). The graph is a connected tree.}
\label{fig:graph}
\end{figure}
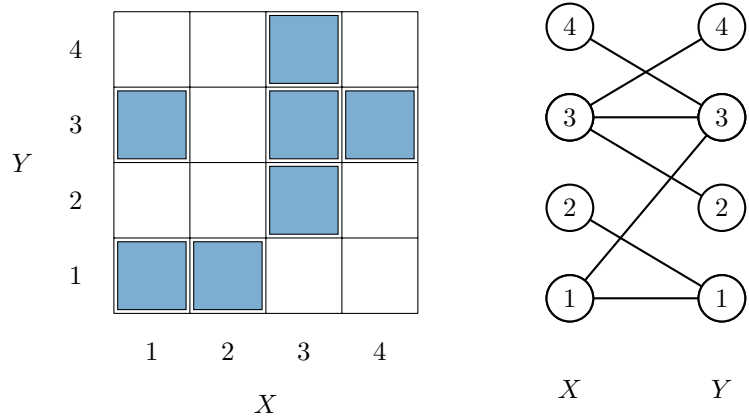

It remains to show that connected equality sets $E(\Phi,\Psi)$ are mccm. Suppose that $\Gamma\supseteq E(\Phi,\Psi)$ is ccm. Then by Lemma \ref{lem:ccm=E} we can find a pricing scheme $(\Phi',\Psi')$ with $E(\Phi',\Psi')\supseteq E(\Phi,\Psi)$. But using just the equalities in (\ref{pricing}) coming from the connected $E(\Phi,\Psi)$, we already find that $\Phi'=\Phi+a$ and $\Psi'=\Psi+a$, so we must have $E(\Phi',\Psi')=E(\Phi,\Psi)$.

(3) is trivial from the proof of (2) that mccm sets are connected.
\end{proof}

\begin{proof}[Proof sketch of Lemma~\ref{lem:ccm=E}]
Our proof will give some additional information on the set of all pricing schemes that satisfy $E(\Phi,\Psi)\supset\Gamma$ and $\Phi(x_0)=0$ for some reference
point $x_0\in X$ to fix the otherwise arbitrary additive constant. Namely we will explicitly construct the largest element $(\Phi_+,\Psi_+)$ of this set and the
smallest $(\Phi_-,\Psi_-)$, so that all other schemes $(\Phi,\Psi)$ satisfy
\begin{equation}\label{orderedprice}
  \Phi_-(x)\leq\Phi(x)\leq\Phi_+(x) \quad\mbox{and}\quad \Psi_-(y)\leq\Psi(y)\leq\Psi_+(y)
\end{equation}
for all $x\in X$ and $y\in Y$. The idea is to optimize the sums of certain costs over paths in $X\cup Y$.

We define a  $\Gamma$-adapted path as a sequence of vertices $z_1,\ldots,z_n\in X\cup Y$ such that the $z_i\in X\Rightarrow(z_i,z_{i+1})\in\Gamma$, and $z_i\in Y\Rightarrow z_{i+1}\in X$. For such a path we define
\begin{equation}\label{pathcost}
  c(z_1,\ldots,z_n)=\sum_{i=1}^{n-1}c(z_i,z_{i+1}),
\end{equation}
with the convention  $c(y,x):=-c(x,y)$  for $x\in X,\ y\in Y$. Then $\Gamma$ is ccm if and only if $c(z_1,\ldots,z_n,z_1)\leq0$ for every $\Gamma$-adapted closed path. This is immediate for cyclic permutations, and follows for more general ones by cycle decomposition. The assertion of Lemma \ref{lem:ccm=E} is trivial if $\Gamma=\emptyset$, so we can pick a point $x_0\in X$ for which some edge $(x_0,y)\in\Gamma$ exists. Then, for any $z\in X\cup Y$, we define, for $z\neq x_0$,
\begin{equation}\label{pathcostpm}
  \chi_+(z):=-\sup c(x_0,\ldots,z) \quad\mbox{and }\ \chi_-(z):=\sup c(z,\ldots,x_0),
\end{equation}
where the suprema are over all $\Gamma$-adapted paths between the specified endpoints, we define $\chi_+(x_0):=\chi_-(x_0):=0$, and empty suprema are defined as $-\infty$. Then $\chi_\pm$ are the maximal and minimal pricing schemes, when written as two functions $\Phi_\pm(x)=\chi_\pm(x)$ and $\Psi_\pm(y)=\chi_\pm(y)$ for $x\in X$ and $y\in Y$.

For proving these assertions, consider paths of the type $(x_0,\ldots,y,x)$. For this to be $\Gamma$-adapted, there is no constraint on the last link, so
\begin{equation}\label{minuschi}
  -\chi_+(y) -c(x,y)\leq -\chi_+(x), \quad\mbox{and\ } \sup_y\bigl\{-\chi_+(y) -c(x,y)\bigr\}=\chi_+(x).
\end{equation}
Here the inequality follows because the adapted paths $x_0\to x$ going via $y$ as the last step are a subclass of all adapted paths and give a smaller supremum. The second statement follows, because for $x\neq x_0$ there has to be some last step from $Y$ to $x$. The inequality \eqref{minuschi} also shows that $(\Phi_+,\Psi_+)$ is a pricing scheme. The same argument applied to the decomposition of paths $(x_0,\ldots,x,y)$ with $(x,y)\in\Gamma$
gives the inequality
\begin{equation}\label{minuschiGamma}
  -\chi_+(x)+c(x,y)\leq -\chi_+(y) \quad\mbox{for\ } (x,y)\in\Gamma.
\end{equation}
Combined with inequality \eqref{minuschi} we get that $(\Phi_+,\Psi_+)$ has equality set $E(\Phi_+,\Psi_+)$ at least $\Gamma$. The corresponding statements for $\chi_-$ follow by first considering paths $(y,x,\ldots,x_0)$ and then $(x,y\ldots,x_0)$ with $(x,y)\in\Gamma$.

Finally, in order to show the inequalities \eqref{orderedprice}, let $(\Phi,\Psi)$ be a tight pricing scheme with $\Phi(x_0)=0$ and $E(\Phi,\Psi)\supset\Gamma$. Consider first any $\Gamma$-adapted path $(x_0,y_0,x_1,\ldots,x_n,y)$. Then,
\begin{eqnarray}\label{fipsitight}
  c(x_0,\ldots,x_n,y)&=&\sum_{i=0}^{n-1}\bigl(\Phi(x_i)-\Psi(y_i)-c(x_{i+1},y_i)\bigl) +\Phi(x_n)-\Psi(y) \nonumber\\
                   &=&\Phi(x_0)-\Psi(y) + \sum_{i=0}^{n-1}\bigl(\Phi(x_{i+1})-\Psi(y_i)-c(x_{i+1},y_i)\bigl) \nonumber\\
                   &\leq&\Phi(x_0)-\Psi(y)=-\Psi(y),
\end{eqnarray}
because the sum is termwise non-positive due to the pricing scheme property. Hence by taking the supremum we get $\chi_+(y)\geq \Psi(y)$. The other inequalities follow with the same arguments applied to paths of the type
$(x_0,\ldots,y_n,x)$, $(x,y_0,\ldots,x_0)$, and $(y,x_1,\ldots,x_0)$.
\end{proof}

Let us summarize the consequences of Proposition~\ref{prop:extremeTransport} for the computation of minimal costs (\ref{primal}). Given any cost function $c$, the first step is to enumerate the corresponding mccm sets, say $\Gamma_\alpha$, $\alpha\in\ops$, for some \emph{finite} label set $\ops$, and to compute for each of these the pricing scheme $(\Phi_\alpha,\Psi_\alpha)$ (up to an overall additive constant, see Proposition \ref{prop:extremeTransport}). This step depends only on the chosen cost function $c$. Then, for any distributions $p,q$ we get
\begin{equation}\label{maxalpha}
  \hat c(p,q)=\check c(p,q)=\max_{\alpha\in\ops}\sum_{x}\Phi_\alpha(x)p(x)-\sum_{y}\Psi_\alpha(y)q(y).
\end{equation}
This is very fast to compute, so the preparatory work of determining the $(\Phi_\alpha,\Psi_\alpha)$ is well invested if many such expressions have to be computed. However, even more important for us that \eqref{maxalpha} simplifies the variational problem sufficiently so that we can combine it with the optimization over joint measurements (see Sect.\ \ref{sec:URM}). Of course, this leaves open the question of how to determine all mccm sets for a cost function. Some remarks about this will be collected in the next subsection.

\subsection{How to find all mccm sets}

We will begin with a basic algorithm for the general finite setting, in which $X,Y$, and the cost function $c$ are arbitrary. Often the task can be greatly simplified if more structure is given. These simplifications will be described in the following sections.

The basic algorithm will be a growth process for ccm subsets $\Gamma\subseteq X\times Y$, which stops as soon as $\Gamma$ is connected (cf.\ the proof of Proposition \ref{prop:extremeTransport}(2)). After that, we can compute the unique pricing scheme $(\Phi,\Psi)$ with equality on $\Gamma$ by solving the system of linear equations with $(x,y)\in\Gamma$ from (\ref{pricing}). This scheme may have additional equality pairs extending $\Gamma$ to an mccm set. Hence, the same $(\Phi,\Psi)$ and mccm sets may arise from another route of the growth process. Nevertheless, we can stop the growth when $\Gamma$ is connected, and eliminate doubles as a last step of the algorithm. The main part of the algorithm will thus aim at finding all {\it connected ccm trees}, where by definition a tree is a graph containing no cycles. We take each tree to be given by a list of edges $(x_1,y_1),\ldots(x_N,y_N)$, which we take to be written in lexicographic ordering, relative to some arbitrary numberings $X=\{1,\ldots,\abs X\}$ and $Y=\{1,\ldots,\abs Y\}$. Hence the first element in the list will be $(1,y)$, where $y$ is the first element connected to $1\in X$.

At stage $k$ of the algorithm we will have a list of all possible initial sequences $(x_1,y_1),\ldots(x_k,y_k)$ of lexicographically ordered ccm trees. For each such sequence the possible next elements will be determined, and all the resulting edge-lists of length $k+1$ form the next stage of the algorithm. Now suppose we have some list
$(x_1,y_1),\ldots(x_k,y_k)$. What can the next pair $(x',y')$ be? There are two possibilities:
\begin{itemize}
\item[(a)] $x'=x_k$ is unchanged. Then lexicographic ordering dictates that $y'>y_k$. Suppose that $y'$ is already connected to some $x<x_k$. Then adding the edge $(x_k,y')$ would imply that $y'$ could be reached in two different ways from the starting node ($x=1$). Since we are looking only for trees, we must therefore restrict to only those $y'>y_k$ which are yet unconnected.
\item[(b)] $x$ is incremented. Since in the end all vertices $x$ must lie in one connected component, the next one has to be $x'=x_k+1$. Since the graphs at any stage should be connected, $y'$ must be a previously connected $Y$-vertex.
\end{itemize}
With each new addition we also check the ccm property of the resulting graph. The best way to do this is to store with any graph the functions $\Phi,\Psi$ on the set of already connected nodes (starting from $\Phi(1)=0$), and update them with any growth step. We then only have to verify inequality \eqref{pricing} for every new node paired with every old one. Since the equality set of any pricing scheme is ccm, this is sufficient.
The algorithm will stop as soon as all nodes are included, i.e., after $\abs X+\abs Y-1$ steps.

\subsection{The linearly ordered case}
When we look at standard quantum observables, given by a Hermitian operator $A$, the outcomes are understood to be the eigenvalues of $A$, i.e., real numbers. Moreover, we typically look at cost functions which depend on the difference $(x-y)$ of two eigenvalues, i.e.,
\begin{equation}\label{chR}
  c(x,y)=h\bigl(x-y\bigr).
\end{equation}
For the Wasserstein distances one uses $h(t)=\abs t^\alpha$ with $\alpha\geq1$. The following Lemma allows, in addition, arbitrary convex, not necessarily even  functions $h$.

\begin{Lemma}
Let $h:\Rl\to\Rl$ be convex, and $c$ be given by \eqref{chR}. Then for $x_1\leq x_2$ and $y_1\leq y_2$ we have
\begin{equation}\label{oneDineq}
  c(x_1,y_1)+c(x_2,y_2)\leq c(x_1,y_2)+c(x_2,y_1),
\end{equation}
with strict inequality if $h$ is strictly convex, $x_1<x_2$ and $y_1<y_2$.
\end{Lemma}

\begin{proof}Since $x_2-x_1\geq0$ and $y_2-y_1\geq0$, there exists $\lambda\in[0,1]$ such that $(1-\lambda)(x_2-x_1)=\lambda(y_2-y_1)$. This implies $x_1-y_1=\lambda(x_1-y_2)+(1-\lambda)(x_2-y_1)$, so that convexity of $h$ gives $c(x_1,y_1)=h(x_1-y_1)\leq\lambda h(x_1-y_2)+(1-\lambda)h(x_2-y_1)=\lambda c(x_1,y_2)+(1-\lambda)c(x_2,y_1)$. The same choice of $\lambda$ also implies $x_2-y_2=(1-\lambda)(x_1-y_2)+\lambda(x_2-y_1)$, so that similarly $c(x_2,y_2)\leq(1-\lambda)c(x_1,y_2)+\lambda c(x_2,y_1)$. Adding up the two inequalities yields the desired result. If $x_1<x_2$ and $y_1<y_2$ are strict inequalities, then $\lambda\in(0,1)$, so that strict convexity of $h$ gives a strict overall inequality.
\end{proof}

As a consequence, if $\Gamma$ is a ccm set for the cost function $c$ and $(x_1,y_1)\in\Gamma$, then all $(x,y)\in\Gamma$ satisfy either $x\leq x_1$ and $y\leq y_1$ or $x\geq x_1$ and $y\geq y_1$. Loosely speaking, while in $\Gamma$, one can only move north-east or south-west, but never north-west or south-east.

This has immediate consequences for ccm sets: In each step in the lexicographically ordered list (see the algorithm in the previous subsection) one either has to increase $x$ by one or increase $y$ by one, going from $(1,1)$ to the maximum. This is a simple drive on the Manhattan grid, and is parameterized by the instructions on whether to go north or east in every step. Of the $\abs X+ \abs Y-2$ necessary steps, $\abs X-1$ have to go in the east direction, so altogether we will have at most
\begin{equation}\label{ManhatCount}
  r=\begin{pmatrix}\abs X+ \abs Y-2\\\abs X-1 \end{pmatrix}
\end{equation}
mccm sets and pricing schemes. They are quickly enumerated without going through the full tree search described in the previous subsection.

\subsection{The metric case}

Another case in which a little bit more can be said is the following \cite[Case~5.4,\ p.56]{Villani}:

\begin{Lemma}\label{lem:metric}
Let $X=Y$, and consider a cost function $c(x,y)$ which is a metric on $X$. Then:\\
\noindent(1)\ Optimal pricing schemes satisfy $\Phi=\Psi$, and the Lipshitz condition $\abs{\Phi(x)-\Phi(y)}\leq c(x,y)$.\\
\noindent(2)\ All mccm sets contain the diagonal.
\end{Lemma}

\begin{proof}
Any pricing schemes satisfies $\Phi(x)-\Psi(x)\leq c(x,x)=0$, i.e., $\Phi(x)\leq\Psi(x)$.
For an optimal scheme, and $y\in X$, we can find $x'$ such that  $\Psi(y)=\Phi(x')-c(x',y)$. Hence
\begin{equation}\label{Lipsh}
  \Psi(y)-\Psi(x)\leq \bigl(\Phi(x')-c(x',y)\bigr)+\bigl(c(x',x)-\Phi(x')\bigr)\leq c(y,x).
\end{equation}
By exchanging $x$ and $y$ we get $\abs{\Psi(y)-\Psi(x)}\leq c(y,x)$. Moreover, given $x$, some $y$ will satisfy
\begin{equation}\label{Phimetric}
  \Phi(x)=\Psi(y)+c(x,y)\geq \Psi(x),
\end{equation}
which combined with the previous first inequality gives $\Phi=\Psi$. In particular, every $(x,x)$ belongs to the equality set.
\end{proof}

One even more special case is that of the discrete metric, $c(x,y)=1-\delta_{xy}$. In this case it makes no sense to look at error exponents, because $c(x,y)^\alpha=c(x,y)$. Moreover, the Lipshitz condition
$\abs{\Phi(x)-\Phi(y)}\leq c(x,y)$ is vacuous for $x=y$, and otherwise only asserts that $\Phi(x)-\Phi(y)\leq1$, which after adjustment of a constant just means that $\abs{\Phi(x)}\leq1/2$ for all $x$.
Hence the transportation cost is just the $\ell^1$ norm up to a factor, i.e.,
\begin{equation}\label{discreteMet}
  \cc(p,q)=\frac12\sup_{\abs\Phi\leq1}\sum_x(p(x)-q(x))\Phi(x)=\frac12\sum_x\abs{p(x)-q(x)}.
\end{equation}


\bibliography{sections/MUR}

\begin{thebibliography}{10}

\bibitem{Appleby}
D.~M. Appleby.
\newblock Concept of experimental accuracy and simultaneous measurements of
  position and momentum.
\newblock {\em Int. J. Theor. Phys.}, 37(5):1491--1509, 1998.

\bibitem{ApplebyNew}
D.~M. Appleby.
\newblock Quantum errors and disturbances: Response to {Busch, Lahti and
  Werner}.
\newblock arXiv:1602.09002, 2016.

\bibitem{BKW}
P.~Busch, J.~Kiukas, and R.~F. Werner.
\newblock Sharp uncertainty relations for number and angle, 2016.
\newblock in preparation.

\bibitem{BLW1}
P.~Busch, P.~Lahti, and R.~F. Werner.
\newblock Proof of {H}eisenberg's {E}rror-{D}isturbance {R}elation.
\newblock {\em Phys. Rev. Lett.}, 111:160405, 2013.
\newblock and arXiv:1306.1565.

\bibitem{BLW2014}
P.~Busch, P.~Lahti, and R.~F. Werner.
\newblock Heisenberg uncertainty for qubit measurements.
\newblock {\em Phys. Rev. A}, 89:012129, 2014.

\bibitem{BLW2}
P.~Busch, P.~Lahti, and R.~F. Werner.
\newblock Measurement uncertainty relations.
\newblock {\em J. Math. Phys.}, 55:042111, 2014.
\newblock and arXiv:1312.4392.

\bibitem{BLW-rms}
P.~Busch, P.~Lahti, and R.~F. Werner.
\newblock Quantum root-mean-square error and measurement uncertainty relations.
\newblock {\em Rev. Mod. Phys.}, 86:1261--1281, 2014.
\newblock and arXiv:.

\bibitem{AMU}
L.~Dammeier, R.~Schwonnek, and R.~F. Werner.
\newblock Uncertainty relations for angular momentum.
\newblock {\em New J. Phys.}, 17:093046, 2015.
\newblock and arXiv:1505.00049.

\bibitem{gb08}
M.~Grant and S.~Boyd.
\newblock Graph implementations for nonsmooth convex programs.
\newblock In V.~Blondel, S.~Boyd, and H.~Kimura, editors, {\em Recent Advances
  in Learning and Control}, Lecture Notes in Control and Information Sciences,
  pages 95--110. Springer-Verlag Limited, 2008.
\newblock \url{http://stanford.edu/~boyd/graph\_dcp.html}.

\bibitem{cvx}
M.~Grant and S.~Boyd.
\newblock {CVX}: Matlab software for disciplined convex programming, version
  2.1.
\newblock \url{http://cvxr.com/cvx}, Mar. 2014.

\bibitem{Heisenberg}
W.~Heisenberg.
\newblock {\"U}ber den anschaulichen {I}nhalt der quantentheoretischen
  {K}inematik und {M}echanik.
\newblock {\em Zeitschr. Phys.}, 43:172--198, 1927.

\bibitem{Holevo}
A.~S. Holevo.
\newblock Statistical decision theory for quantum systems.
\newblock {\em J. Multivariate Anal.}, 3:337, 1973.

\bibitem{Kennard}
E.~Kennard.
\newblock Zur {Q}uantenmechanik einfacher {B}ewegungstypen.
\newblock {\em Zeitschr. Phys.}, 44:326--352, 1927.

\bibitem{Nikaido}
H.~Nikaid{\^o}.
\newblock On von {N}eumann's minimax theorem.
\newblock {\em Pacific J. Math.}, 4:65--72, 1954.

\bibitem{Ozawa}
M.~Ozawa.
\newblock Uncertainty relations for joint measurements of noncommuting
  observables.
\newblock {\em Phys. Lett. A}, 320:367--374, 2004.

\bibitem{OzawaDis}
M.~Ozawa.
\newblock Disproving {H}eisenberg's error-disturbance relation.
\newblock 2013.
\newblock arXiv:1308.3540.

\bibitem{rockafellar}
R.~T. Rockafellar.
\newblock {\em Convex Analysis}.
\newblock Princeton University Press, Princeton, 1970.

\bibitem{Sion}
M.~Sion.
\newblock On general minimax theorems.
\newblock {\em Pac. J. Math}, 8:171–176.

\bibitem{boyd-vandenberghe}
L.~Vandenberghe and S.~Boyd.
\newblock Semidefinite programming.
\newblock {\em SIAM Rev.}, 38:49–95, 1996.

\bibitem{boyd-vandenberghe-book}
L.~Vandenberghe and S.~Boyd.
\newblock {\em Convex Optimization}.
\newblock Cambridge UP, 2004.

\bibitem{Villani}
C.~Villani.
\newblock {\em Optimal Transport: Old and New}.
\newblock Springer, 2009.

\bibitem{QHA}
R.~F. Werner.
\newblock Quantum harmonic analysis on phase space.
\newblock {\em J. Math. Phys.}, 25:1404--1411, 1984.

\bibitem{Wer04}
R.~F. Werner.
\newblock The uncertainty relation for joint measurement of position and
  momentum.
\newblock {\em Quant. Inform. Comput.}, 4:546--562, 2004.
\newblock and arXiv:quant-ph/0405184.

\bibitem{WerFou}
R.~F. Werner.
\newblock Uncertainty relations for general phase spaces.
\newblock {\em Front. Phys.}, 11:110305, 2016.
\newblock proceedings of the QCMC 2014, and arXiv:1601.03843.

\bibitem{Yuen}
H.~P. Yuen, R.~S. Kennedy, and M.~Lax.
\newblock Optimum testing of multiple hypotheses in quantum detection theory.
\newblock {\em IEEE Trans. Inf. Theory}, IT-21.

\end{thebibliography}

\end{document}